\documentclass[11pt,letterpaper]{article}
\usepackage{fullpage,titling}
\usepackage{url}
\usepackage[utf8]{inputenc}
\usepackage[small]{caption}
\usepackage[whole]{bxcjkjatype}
\usepackage[linesnumbered,ruled,vlined]{algorithm2e}

\usepackage{array,multirow,multicol,graphicx,booktabs,bm,diagbox,vwcol}
\usepackage[svgnames]{xcolor}
\usepackage{tikz}
\usepackage{pgfplots}
\usepackage{pgfplotstable}
\usetikzlibrary{positioning,intersections,calc}
\pgfplotsset{compat=newest}
\usepackage{url}
\usepackage[colorinlistoftodos]{todonotes}
\usepackage{subcaption}

\usepackage[whole]{bxcjkjatype}
\usepackage[linesnumbered,ruled,vlined]{algorithm2e}
\usepackage{amsmath,amssymb,amsthm,mathtools}
\newtheorem{theorem}{Theorem}

\newtheorem{proposition}{Proposition}
\newtheorem{corollary}{Corollary}
\newtheorem{definition}{Definition}

\newtheorem{example}{Example}

\usepackage{stackengine}
\newcommand\mm[1]{\,\stackengine{.9ex}{}{\tiny\textcolor{blue}{#1}}{O}{l}{F}{F}{L}}

\newcommand{\ot}{\leftarrow}
\newcommand{\argmin}{\mathop{\rm arg\,min}}
\newcommand{\argmax}{\mathop{\rm arg\,max}}
\newcommand{\OPT}{\mathrm{OPT}}
\newcommand{\ALG}{\mathrm{Alg}}

\newcommand{\cl}{\mathop{\mathrm{cl}}\nolimits}
\newcommand{\cI}{\mathcal{I}}
\newcommand{\cU}{\mathcal{U}}
\newcommand{\StP}{$\mathrm{\Sigma_{2}^{P}}$}
\newcommand{\EASSP}{\textbf{$\boldsymbol{\exists\forall}$\texttt{SubsetSum}}}
\newcommand{\EAtDM}{\textbf{$\boldsymbol{\exists\forall}$\texttt{3DM}}}
\newcommand{\DMP}{\textbf{\texttt{DisjointMatching}}}

\newcommand{\Usub}{\cU_{\mathrm{sub}}}
\newcommand{\Uadd}{\cU_{\mathrm{add}}}
\newcommand{\Ucard}{\cU_{\mathrm{card}}}
\newcommand{\Gcap}{\Gamma_{\mathrm{cap}}}
\newcommand{\Gmat}[1]{\Gamma^{(#1)}_{\mathrm{mat}}}
\newcommand{\Gvec}[1]{\Gamma^{(#1)}_{\mathrm{knap}}}

\title{Approximately Stable Matchings with General Constraints}
\author{{Yasushi Kawase}\\[5pt]
  {\normalsize Tokyo Institute of Technology and} \\
  {\normalsize RIKEN AIP Center, Tokyo, Japan.}\\
  {\normalsize \texttt{kawase.y.ab@m.titech.ac.jp}}
  \and 
  {Atsushi Iwasaki}\\[5pt]
  {\normalsize University of Electro-Communications and} \\
  {\normalsize RIKEN AIP Center, Tokyo, Japan.}\\
  {\normalsize \texttt{iwasaki@is.uec.ac.jp}}
}
\date{}

\begin{document}
\begin{titlingpage}
\maketitle
\begin{abstract}
This paper focuses on two-sided matching
where one side (a hospital or firm) is matched to the other side (a doctor or worker) 
so as to maximize a cardinal objective under general feasibility constraints.
In a standard model, even though multiple doctors can be matched to a single hospital, 
a hospital has a \emph{responsive preference} and  a \emph{maximum quota}.
However, in practical applications, a hospital has some complicated cardinal preference and constraints. 
With such preferences (e.g., submodular) and constraints (e.g., knapsack or matroid intersection), stable matchings may fail to exist. 
This paper first determines the complexity of checking and computing stable matchings based on preference class and constraint class.
Second, we establish a framework to analyze this problem on \emph{packing} problems and
the framework enables us to access the wealth of online packing algorithms
so that we construct \emph{approximately stable} algorithms as a variant of generalized deferred acceptance algorithm. 
We further provide some inapproximability results. 
\end{abstract}
\end{titlingpage}

\section{Introduction}
This paper studies a two-sided, one-to-many matching model 
when one side (a hospital or firm) is allocated members from the other side (a doctor or worker), 
covering constraints to satisfy practical or social demands and prohibiting infeasible allocation (matching). 
The theory of two-sided matching has been extensively developed, 
as illustrated by the comprehensive surveys~\cite{Roth:CUP:1990,manlove:2013}.

Matching with constraints has been prominent across computer science and economics since the seminal work by~\cite{kamakoji-basic}. 
In many applications, various constraints are often imposed on an outcome, 
e.g., \emph{type-specific quotas} on hospitals to assign several different types (skills) of doctors~\cite{Abdulkadiroglu:AER:2003},
\emph{budget constraints} on hospitals to limit the total amount of wages~\cite{abizada:te:2016,KI2017,KI2018}.
The current paper exactly covers these complicated constraints and further generalizes them. 
Specifically, we consider general constraints of upper bounds known as \emph{independence system constraints}, i.e., any subset of a feasible set of doctors is also feasible.
We assume that each constraint is represented by a \emph{capacity} (maximum quota), an \emph{intersection of multiple matroids}, or a \emph{multi-dimensional knapsack}.
It should be noted that every independence system constraint can be represented by an intersection of multiple matroids and a multi-dimensional knapsack.
In addition, we assume that each hospital's preference is represented by a utility function.
We consider three important classes of cardinal utilities: \emph{cardinality}, \emph{additive}, and \emph{submodular}.

With such preferences and constraints, stable matchings may fail to exist.
Determining whether a given instance has a stable matching is hard in general. 
It is known to be \StP{}-complete when hospitals' utilities are additive, and the constraints are given as ($1$-dimensional) knapsack~\cite{HISY2017}. 
Note that the existence of stable matchings is guaranteed 
when the utilities are additive and the constraints are matroid~\cite{kty:jet:2018}
or the utilities are cardinality (matched size) and the constraints are knapsack~\cite{KI2017}.

There are several possibilities to circumvent the nonexistence problem. 
One modifies the notion of the stability and proposes a variant of the Deferred Acceptance (DA) algorithm~\cite{hafalir2013effective,kamakoji-basic,abizada:te:2016}.
Another restricts hospitals' priorities to ensure the existence of a stable matching, e.g., lexicographic priorities~\cite{dean:ifip:2006}.
Alternatively, Kawase and Iwasaki~\cite{KI2017} and Nguyen \textit{et al.}~\cite{NNT2019} focused on
\textit{near-feasible} stable matchings that approximately satisfy each budget of the hospitals.

This paper focuses on \emph{approximately stable} matchings where
the participants are only willing to change the assignments for a multiplicative improvement of a certain amount~\cite{arkin2009}.
This idea can be interpreted as one in which
a hospital in a blocking pair changes its match as soon as its utility after the change increases by any (arbitrarily small) amount.
Arkin \textit{et al.}~\cite{arkin2009} examined a stable roommate problem, which is a non-bipartite one-to-one matching problem, 
while we examine a bipartite one-to-many matching problem.
It is reasonable for a hospital to change its assignments only in favor of a significant improvement; even though the grass may be greener on the other side, crossing the fence takes effort.

\medskip
\noindent\textbf{Our results}\quad
First, we analyze the problem of checking the stability
on (offline) \emph{packing problems} so that we understand the features and obtain the complexity results, which vary according to hospitals' utilities and imposed constraints. In particular, Theorem~\ref{thm:check} proves that given a matching, checking whether it is stable or not is equivalent to solving a packing problem.
Once we know the complexity of a packing problem in a given setting (utilities and constraints), we obtain that of the checking problem associated with the setting. 
Our results are summarized in Table~\ref{tab:summary}\,(\subref{stab:stab}). 

Second, Table~\ref{tab:summary}\,(\subref{stab:exist}) summarizes our trichotomy results characterizing the complexity of determining whether a given instance has a stable matching or not.
The problem is polynomially solvable in very restricted classes of utilities and constraints, while it is NP-complete or \StP{}-complete in the other settings. 
Here, \StP{} (also known as $\mathrm{NP^{NP}}$) is the class of problems solvable in polynomial time by a nondeterministic Turing machine with an oracle for some NP-complete problem.
To prove the NP-hardness, we give a reduction from \DMP{}, which is NP-complete~\cite{frieze1983}.
In addition, we prove the \StP{}-hardness by reductions from the \EAtDM{} or \EASSP{}, which are \StP-complete~\cite{berman1997,mcloughlin1984}.

Finally, we introduce a framework that leads us to construct algorithms that find approximately stable matchings as a variant of generalized deferred acceptance (GDA). 
Intuitively, each hospital utilizes an \emph{online} packing algorithm while running a GDA procedure.
By applying the property called \emph{$\alpha$-approximation}~\cite{KI2018}, we show that if there exists an $\alpha$-competitive algorithm for a class of online packing problems, then we can construct a DA algorithm that always yields an $\alpha$-stable matching for the markets in a corresponding class. 
Table~\ref{tab:summary}\,(\subref{stab:apx}) summarizes the upper and lower bounds of the approximation ratios that we obtained. Note that Theorem~\ref{thm:general_lower} provides a basis to derive some novel lower bounds. 
Here, for the knapsack constraints, we assume that the weight of each element on every dimension is at most a $(1-\epsilon)$ fraction of the total capacity.

\begin{table}[ht]
  \caption{Summary of results ($k\ge 3$ and $\rho\ge 2$)}\label{tab:summary}
  {\centering
  \begin{subtable}{.49\linewidth}
    \subcaption{Complexity of checking the stability}\label{stab:stab}
    \scalebox{0.92}{%
      {\renewcommand\arraystretch{1.1}
        \setlength{\tabcolsep}{4pt}
        \begin{tabular}[htbp]{p{22mm}|lll}\toprule
          \diagbox[width=24mm]{\scriptsize Constraints}{\scriptsize Hosp.\ Utils}& Cardinality & Additive & Submodular \\ \midrule
          Capacity          & P          & P          & coNP-c\\\hline
          Matroid           & P　　    　& P          & coNP-c\\
          2-mat.\ int.    & P          & P          & coNP-c\\
          $k$-mat.\ int.  & coNP-c     & coNP-c     & coNP-c\\\hline
          1-dim.\ knap.     & P          & coNP-c\mm{$\dagger$} & coNP-c\\
          $\rho$-dim.\ knap.& coNP-c     & coNP-c     & coNP-c\\
          \bottomrule
        \end{tabular}}}
  \end{subtable}%
  \begin{subtable}{.02\linewidth}\mbox{}\end{subtable}%
  \begin{subtable}{.49\linewidth}
    \subcaption{Complexity of checking existence}\label{stab:exist}
    \scalebox{0.92}{%
      {\renewcommand\arraystretch{1.1}
        \setlength{\tabcolsep}{4pt}
        \begin{tabular}[htbp]{p{22mm}|lll}\toprule
        \diagbox[width=24mm]{\scriptsize Constraints}{\scriptsize Hosp.\ Utils}& Cardinality & Additive & Submodular \\ \midrule
          Capacity          & P\mm{*}          & P\mm{*}                & \StP-c\mm{b}\\\hline
          Matroid           & P\mm{*}          & P\mm{*}                & \StP-c\mm{b}\\
          2-mat.\ int.    & NP-c\mm{a}       & NP-c\mm{a}             & \StP-c\mm{b}\\
          $k$-mat.\ int.  & \StP-c\mm{c}     & \StP-c\mm{c}           & \StP-c\mm{b,c}\\\hline
          1-dim.\ knap.     & P\mm{$\ddagger$} & \StP-c\mm{$\dagger$}   & \StP-c\mm{$\dagger$,b}\\
          $\rho$-dim.\ knap.& \StP-c\mm{d}     & \StP-c\mm{$\dagger$,d} & \StP-c\mm{$\dagger$,b,d}\\
          \bottomrule
        \end{tabular}}}
  \end{subtable}\\[20pt]}
  \begin{subtable}{\linewidth}
    \centering
    \subcaption{Approximation ratios (Upper Bound$\bigm/$Lower Bound)}\label{stab:apx}
      {\renewcommand\arraystretch{1.1}
        \setlength{\tabcolsep}{9pt}
        \begin{tabular}[htbp]{p{30mm}|lll}\toprule
          \diagbox[width=32mm]{\scriptsize Constraints}{\scriptsize Hosp.\ Utils}& Cardinality & Additive & Submodular \\ \midrule
          Capacity          & $1\mm{*} \bigm/ 1$ & $1\mm{*} \bigm/ 1$  & $4\mm{g} \bigm/ 1.28\mm{l}$\\\hline
          Matroid           & $1\mm{*} \bigm/ 1$ & $1\mm{*} \bigm/ 1$  & $4\mm{g} \bigm/ 1.28\mm{l}$\\
          2-matroid int.    & $2\mm{e} \bigm/ 2\mm{k}$    & $(\!\sqrt{2}{+}1)^2\mm{f} \bigm/ 2\mm{k}$ & $8\mm{g} \bigm/ 2\mm{k}$\\
          $k$-matroid int.  & $k\mm{e} \bigm/ 2\mm{k}$    & \scalebox{0.9}{$(\!\!\sqrt{k}{+}\sqrt{k{-}1}\!)^2\mm{f} \bigm/ k\mm{m}$}  & $4k\mm{g} \bigm/ k\mm{m}$\\\hline
          1-dim.\ knap.     & $1\mm{$\ddagger$} \bigm/ 1$    & $\frac{1}{\epsilon}\mm{$\ddagger$} \bigm/ \frac{1}{\epsilon}\mm{$\ddagger$}$ & $O(\frac{1}{\epsilon^2})\mm{j} \bigm/ \frac{1}{2\epsilon}\mm{n}$\\
          $\rho$-dim.\ knap.& $\rho\mm{h} \bigm/ 2\mm{k}$ & $\frac{\rho}{\epsilon}\mm{i} \bigm/ \frac{\rho}{2\epsilon}\mm{n}$  & $O(\frac{\rho}{\epsilon^2})\mm{j} \bigm/ \frac{\rho}{2\epsilon}\mm{n}$\\
          \bottomrule
        \end{tabular}}
  \end{subtable}\\[5pt]
  {\footnotesize \mm{*}~\cite{kty:jet:2018}; \mm{$\dagger$}~\cite{HISY2017};
  \mm{$\ddagger$}~\cite{KI2017};
  \mm{a}~Thm.~\ref{thm:2hard}; \mm{b}~Thm.~\ref{thm:Sigma2Pa}; \mm{c}~Thm.~\ref{thm:Sigma2Pb}; \mm{d}~Thm.~\ref{thm:Sigma2Pc};
  \mm{e}~Cor.~\ref{thm:cardmat}; \mm{f}~Cor.~\ref{thm:addmat}; \mm{g}~Cor.~\ref{thm:submat}; \mm{h}~Cor.~\ref{thm:cardknap}; \mm{i}~Cor.~\ref{thm:addknap}; \mm{j}~Cor.~\ref{thm:subknap};
  \mm{k}~Ex.~\ref{ex:nonexistence}; \mm{l}~Ex.~\ref{ex:submo};
  \mm{m}~Thm.~\ref{thm:lower_addmat}; \mm{n}~Thm.~\ref{thm:lower_addknap}.
  }
\end{table}

\section{Model}
This section describes our model of two-sided matching markets. 
A market is a tuple $(D,H,{\succ_D},\allowbreak u_H,\cI_H)$ where each component is defined as follows. 
There is a finite set of doctors $D=\{d_1,\ldots,d_n\}$ and a finite set of hospitals $H=\{h_1,\ldots,h_m\}$.
We denote by $\succ_D={(\succ_d)}_{d\in D}$ the doctors' preference profile 
where $\succ_d$ is the strict relation of $d\in D$ over $H\cup\{\emptyset\}$; 
$s\succ_d t$ means that $d$ strictly prefers $s$ to $t$, where $\emptyset$ denotes being unmatched.
Let $u_H={(u_h)}_{h\in H}$ denote the hospitals' cardinal preference profile where $u_h$ is the \emph{utility function} $u_h\colon 2^{D}\to\mathbb{R}_+$.
We assume that $u_h$ is \emph{normalized} (i.e., $u_h(\emptyset)=0$) and \emph{monotone} (i.e., $u_h(D'')\le u_h(D')$ for any $D''\subseteq D'\subseteq D$).
Let $\cI_H=(\cI_h)_{h\in H}$ denote the \emph{feasibility constraints} for hospitals where $\cI_h\subseteq 2^{D}$ for each $h\in H$.
We assume that $(D,\cI_h)$ is an \emph{independence system} for each $h\in H$, i.e.,
(I1) \(\emptyset\in\cI_h\) and
(I2) \(S\subseteq T\in\cI_h\) implies \(S\in\cI_h\).
Here, $\cI_h$ is called an \emph{independence family}.
We say that $h \in H$ is \emph{acceptable} to $d \in D$ if $h \succ_{d} \emptyset$.
In addition, $D' \subseteq D$ is said to be \emph{feasible} to $h \in H$ if $D'\in\cI_h$.

Note that a pair of utility $u_h$ and constraint $\cI_h$ can be represented by a single (non-monotone) utility function 
$\hat{u}_h\colon 2^D\to\mathbb{R}\cup\{-\infty\}$ such that $\hat{u}_h(X)=u_h(X)$ if $X\in\cI_h$ and $\hat{u}(X)=-\infty$ otherwise.
However, we treat it separately to define classes of markets clearly.

For a set of utility functions $\cU$ and a set of independence families $\Gamma$,
a market $(D,H,\succ_D,u_H,\cI_H)$ is $(\cU,\Gamma)$-market if $u_h\in\cU$ and $\cI_h\in\Gamma$ for all $h\in H$.
Namely, the $(\cU,\Gamma)$-markets are those in which the utilities and the feasibility constraints are restricted to be in $\cU$ and $\Gamma$, respectively.
We analyze the properties of the $(\cU,\Gamma)$-markets based on utility class $\cU$ and constraint class $\Gamma$.

A \emph{matching} is a set of pairs 
$\mu\subseteq \{(d,h)\in D\times H\mid h\succ_d\emptyset\}$
such that each doctor appears in at most one pair of $\mu$; that is,
we have $\bigl|\{(d',h')\in\mu \mid d'=d\}\bigr|\le 1$ 
for any $d\in D$.
For $d\in D$, $h\in H$, and a matching $\mu\subseteq D\times H$, we define $\mu(h)\coloneqq \{d'\mid (d',h)\in\mu\}~(\subseteq D)$ and $\mu(d)\coloneqq s~(\in H\cup\{\emptyset\})$ where $s\in H$ if $(d,s)\in\mu$ and $s=\emptyset$ otherwise.

We call matching $\mu$ \emph{feasible} if $\mu(h)\in \cI_h$ for all $h\in H$.
Given a matching $\mu$ and a real $\alpha\ge 1$,
a set of doctors $D'\in\cI_h$ is an \emph{$\alpha$-blocking coalition} for hospital $h$ if
(i) $h\succeq_{d} \mu(d)$ for any $d\in D'$ and (ii) $u_h(D')>\alpha\cdot u_h(\mu(h))$.%
\footnote{Although the second condition can also be defined in an additive manner: $u(X'')>u(X_h')+\alpha$, scale invariance should be required. 
For example, when a market has no $\alpha$-stable matching in the additive sense, a market with the hospitals' utilities that are multiplied by $100$ has no $100\alpha$-stable matching.}
We then obtain a stability concept. 
\begin{definition}[\cite{arkin2009}]\label{def:stability}
  A feasible matching $\mu$ is \emph{$\alpha$-stable} if
  there exists no $\alpha$-blocking coalition. 
\end{definition}
Note that $1$-stability is equivalent to the standard stability concept.
As we will see in Examples~\ref{ex:nonexistence} and~\ref{ex:submo}, $1$-stable matching may not exist in general.
Intuitively, $\alpha$-stable means that the hospitals are only willing to change the assignments for a multiplicative improvement of $\alpha$. 
This idea regards the value of $\alpha$ as a switching cost for the hospitals.

\subsection{Classes of Utilities and Constraints}\label{subsec:class}
Here, we formally describe three important classes of utility functions: \emph{cardinality}, \emph{additive}, and \emph{submodular}, which capture wide varieties of applications. We assume that utility functions are monotone and nonnegative throughout this paper.
First, a utility function $u\colon 2^D\to\mathbb{R}_+$ is called \emph{cardinality} if $u(D')=|D'|$ for all $D'\subseteq D$. 
Let us denote $\Ucard$ as a set of cardinality utility functions. 
Second, it is called \emph{additive} (or modular) if $u(D')=\sum_{d\in D'}u(d)$ holds for all $D'\subseteq D$ (where we denote $u(\{d\})$ by $u(d)$ for simplicity). 
Third, it is called \emph{submodular} if $u(D')+u(D'')\ge u(D'\cup D'')+u(D'\cap D'')$ holds for all $D',D''\subseteq D$ (see \cite{fujishige2005} for more details). 
As well as for the cardinality functions, we define the set of additive and submodular utilities as $\Uadd$ and $\Usub$, respectively. 
Here, $\Ucard\subsetneq \Uadd\subsetneq \Usub$ inevitably holds. 

Next, we formally define three classes of constraints: \emph{capacity}, \emph{matroid intersection}, and \emph{multidimensional knapsack}. 
An independence system \((D,\cI)\) represents a \emph{capacity} constraint of rank $r$ if $\cI=\{D'\subseteq D\mid |D'|\le r\}$. 
We define the set of independence families that represent rank $r$ capacities as $\Gcap^{(r)}$.
Also, we denote $\Gcap$ as $\bigcup_{r\in\mathbb{Z}_+}\Gcap^{(r)}$. 
This class represents a standard matching model with maximum quotas. 

An independence system \((D,\cI)\) is called \emph{matroid} if, 
for \(D',D''\in\cI\), \(|D'|<|D''|\) implies the existence of \(d\in D''\setminus D'\) such that \(D'\cup\{d\}\in\cI\). 
Moreover, it is called \emph{$k$-matroid intersection} if there exist $k$ matroids $(D,\cI^{1}),\dots,(D,\cI^{k})$ such that $\cI=\bigcap_{i\in[k]}\cI^{i}$,
where $[k]$ denotes set $\{1,\dots,k\}$.
We denote the set of independence families of the $k$-matroid intersection as $\Gmat{k}$.

Note that we have $\Gcap\subsetneq \Gmat{1}\subsetneq \Gmat{2}\subsetneq\cdots$.
We assume that each independence system $(D,\cI)$ in $\Gmat{k}$ is represented by $\bigcap_{i\in[k]} \cI^i$ with matroids $(D,\cI^i)$ ($i\in[k]$),
and every $\cI^i$ $(i\in[k])$ is given as a compact representation.
For more details on matroids, see, e.g., \cite{oxley1992mt}.

Furthermore, for a natural number $\rho$ and a positive real $\epsilon$, the set of \emph{$\rho$-dimensional knapsack with $\epsilon$-slack} $\Gvec{\rho,\epsilon}$
is defined as the set of independence families \(\cI\) that can be represented as \[\cI=\big\{D'\subseteq D\bigm| \textstyle\sum_{d\in D'}w(d,i)\le 1~\text{for all}\ i\in[\rho]\big\}\] with weights $w(d,i)\in[0,1-\epsilon]$ for each $d\in D$ and~$i\in[\rho]$.
We assume that independence systems $(D,\cI)$ in $\Gvec{\rho,\epsilon}$ are given by weights.

Note that every independence system can be represented by a matroid intersection and a multidimensional knapsack. 
The representability of matroid intersection is not stronger than that of multidimensional knapsack, and vice versa. 
Formally, $\Gmat{1}\not\subseteq\Gvec{\rho,\epsilon}$ and $\Gvec{1,\epsilon}\not\subseteq\Gmat{k}$ for any positive integers $\rho$ and $k$ and any nonnegative real $\epsilon<1$.

For independence system $(D,\cI)$ and subset $A\subseteq D$,
the \emph{restriction} of $(D,\cI)$ to $A$ is defined as $\cI|A\coloneqq \{X \mid A\supseteq X\in \cI\}$.
In this paper, we only consider a constraint class $\Gamma$ that is closed under the restriction, i.e.,  $\cI\in\Gamma$ implies $\cI|D'\in\Gamma$ for all $D'\subseteq D$.
We remark that $\Gcap^{(r)}$, $\Gmat{k}$, and $\Gvec{\rho,\epsilon}$ satisfy the condition.

\subsection{Applications}\label{subsec:applications}
This section illustrates several existing and critical situations raised in the literature of matchings with constraints 
and describes how our constraint representation (the feasible subsets of doctors) is reduced to such situations.

\smallskip\noindent\textbf{Type-specific quotas}
One of the simplest examples of the feasibility family is \emph{type-specific quotas}, in which
doctors are partitioned based on their types, and each hospital has type-specific quotas in addition to its capacity~\cite{Abdulkadiroglu:AER:2003}.
Fix hospital $h$ and
suppose that $(D_t)_{t\in T}$ is the partition of doctors with types $T$, i.e., $\bigcup_{t\in T}D_t=D$ and $D_t\cap D_{t'}=\emptyset$ for all $t,t'\in T$ with $t\ne t'$.
Let $q\in\mathbb{Z}_+$ be the capacity of $h$, and let $q_t\in\mathbb{Z}_+$ be the quota for type $t\in T$.
Then $D'~(\subseteq D)$ belongs to $\cI_h$ if and only if $|D'|\le q$ and $|D'\cap D_t|\le q_t$ for every $t\in T$.
In this case, $(D,\cI_h)$ is a matroid, and if the utilities are additive, 
a $1$-stable matching always exists and can be found efficiently. 
If the utilities are submodular, the matching no longer exists (see Example~\ref{ex:submo}).
However, we reveal that a $4$-stable matching always exists and is efficiently found (Theorem~\ref{thm:submat}).

\smallskip\noindent\textbf{Overlapping types}
Here, we generalize the type-specific quotas to those where each doctor can simultaneously belong to multiple types~\cite{kurata:jair:2017}.
Fix hospital $h$ and let $T_1,\dots,T_k$ be the sets of types.
For $i\in[k]$, let $(D^i_t)_{t\in T_i}$ be the partition of doctors with types $T_i$.
In addition, let $q\in\mathbb{Z}_+$ be the capacity of $h$, and let $q^i_t\in\mathbb{Z}_+$ be the quota for type $t\in T_i$.
Then $D'~(\subseteq D)$ belongs to $\cI_h$ if and only if $|D'|\le q$ and $|D'\cap D^i_t|\le q^i_t$ for every $i\in[k]$ and $t\in T_i$.
In this case, $(D,\cI_h)$ is a $k$-matroid intersection.
Kurata \emph{et al.}~\cite{kurata:jair:2017} treated quotas as soft constraints that can be violated and found a quasi-stable matching in a different manner.
To the best of our knowledge, we are the first to treat them as hard constraints that should be precisely satisfied and for finding an approximately stable matching.

\smallskip\noindent\textbf{Budget constraints}
Under budget constraints, one side (a firm or hospital) can make monetary transfers (offer wages) to the other (a worker or doctor), and each hospital has a {\em fixed budget}; that is, the total amount of wages allocated by each hospital to doctors is constrained~\cite{abizada:te:2016,KI2017,KI2018}. Let $w^h(d)$ be the offered wage from hospital $h$ to doctor $d$, and let $b_h$ be its budget. Then constraint $\cI_h$ is defined as $\cI=\{D'\subseteq D\mid \sum_{d\in D'}w^h(d)\le b^h\}$ and becomes a $1$-dimensional knapsack. In fact, we have $\cI_h\in\Gvec{1,\epsilon}$ with $\epsilon=1-\max_{d\in D}w^h(d)/b^h$.
Kawase and Iwasaki~\cite{KI2018} considered up to the additive utility case and we disentangled the submodular utility case.

\smallskip\noindent\textbf{Refugee match (multiple resource constraints)}
In refugee resettlement, different refugee families require such various services as school seats, hospital beds, slots in language classes, and employment training programs~\cite{DKT2016}.
Suppose that the set of services is $\Sigma$ and the capacity of $h$ (local areas in this context) is $b_s^h$ for each $s\in\Sigma$.
In addition, each doctor (refugee family) $d$ needs $w^h(d,s)$ units of service $s\in\Sigma$.
Then, the feasibility constraint    of hospital $h$ is defined: $\cI_h=\{D'\subseteq D\mid \sum_{d\in D'}w^h(d,s)\le b^h_s~(\forall s\in\Sigma)\}$.
In this case, the constraint is a $|\Sigma|$-dimensional knapsack.
In fact, we have $\cI_h\in\Gvec{|\Sigma|,\epsilon_h}$ with $\epsilon_h=1-\max_{s\in\Sigma}\max_{d\in D}w^h(d,s)/b^h_s$.

\subsection{Market without Stable Matchings} 
Let us show a market may not have $1$-stable matchings, even when the utilities are cardinality. 
\begin{example}\label{ex:nonexistence}
Consider a market with four doctors $D\coloneqq\{d_1,d_2,d_3,d_4\}$ and two hospitals $H\coloneqq\{h_1,h_2\}$. 
The preferences of the doctors are $h_1\succ_{d_i}h_2\succ_{d_i}\emptyset$ for $i=1,2$ and $h_2\succ_{d_i}h_1\succ_{d_i}\emptyset$ for $i=3,4$.
Suppose that each hospital has a cardinality utility.
The feasibility families are $\cI_{h_1}\coloneqq 2^{\{d_1,d_3\}}\cup 2^{\{d_2,d_4\}}$ and $\cI_{h_2}\coloneqq 2^{\{d_1,d_4\}}\cup 2^{\{d_2,d_3\}}$.
Then, it is straightforward to see that, for $1\le \alpha<2$, this market has no $\alpha$-stable matching (see Proposition~\ref{prop:nonexistence} in Appendix for the formal proof).

Let us remark that the independence systems $(D,\cI_{h_1})$ and $(D,\cI_{h_2})$ can be represented by a 2-matroid intersection and a 2-dimensional knapsack with $\epsilon~(<1/2)$. 
For example, $\cI_{h_1}$ is in $\Gmat{2}$ because $\cI_{h_1}=\cI^1\cap\cI^2$ for
  \begin{align*}
  \cI^1=\Bigl\{\hat{D}\subseteq D \Bigm| \substack{|\hat{D}\cap\{d_1,d_2\}|\le 1,\\ |\hat{D}\cap\{d_3,d_4\}|\le 1}\Bigr\}\quad\text{and}\quad
  \cI^2=\Bigl\{\hat{D}\subseteq D \Bigm| \substack{|\hat{D}\cap\{d_1,d_4\}|\le 1,\\ |\hat{D}\cap\{d_2,d_3\}|\le 1}\Bigr\}.
  \end{align*}
  Further, $\cI_{h_1}$ is in $\Gvec{2,\epsilon}$ because it is represented by the following weights: 
  \begin{align*}
  \begin{array}{llll}
  w(d_1,1)=1-\epsilon,& w(d_2,1)=1/2,& w(d_3,1)=0,         & w(d_4,1)=1/2,\\
  w(d_1,2)=0,         & w(d_2,2)=1/2,& w(d_3,2)=1-\epsilon,& w(d_4,2)=1/2.
  \end{array}
  \end{align*}
\end{example}
Moreover, if hospitals' utilities are submodular, a market fails to have $1$-stable matchings even under capacity constraints.  
\begin{example}\label{ex:submo}
  Consider a market with four doctors $D\coloneqq\{d_1,d_2,d_3,d_4\}$ and two hospitals $H\coloneqq\{h_1,h_2\}$. 
  The preferences of the doctors are 
  $h_1\succ_{d_i}\emptyset\succ_{d_i}h_2$ for $i=1,2$,
  $h_2\succ_{d_3}h_1\succ_{d_3}\emptyset$, and
  $h_1\succ_{d_4}h_2\succ_{d_4}\emptyset$.
  Suppose the $\cI_{h_1}$ and $\cI_{h_2}$ are capacity constraints of rank $2$ and rank $1$, respectively.
  Let $u_{h_1}$ be a submodular utility such that
  \(u_{h_1}(D')\coloneqq \sum_{e\in \bigcup_{d_i\in D'}A_i}w(e)\) where
  $A_1=\{a_{1},a_{3}\}$, $A_2=\{a_{2},a_{4}\}$, $A_3=\{a_{3},a_{4},a_{5}\}$, $A_4=\{a_{1},a_{2}\}$,
  $w(a_1)=w(a_2)=w(a_5)=4,$ and $w(a_3)=w(a_4)=\sqrt{17}-1$.
  Here, $u_{h_1}$ is clearly submodular since it is a weighted-coverage function. 
  Let $u_{h_2}$ be an additive utility such that $u_{h_2}(d_3)=1$ and $u_{h_2}(d_4)=2$.
  Then, it is straightforward to see that, there exists no $(1+\sqrt{17})/4~(\approx 1.28)$-stable matching in this market (see Proposition~\ref{prop:submo} in Appendix for the formal proof).
\end{example}

\section{Checking the Stability of a Given Matching}
In this section, we discuss the computational complexity of checking the $\alpha$-stability of a given matching.
We are going to prove that the problem is equivalent to computing an offline \emph{packing problem} which finds an outcome or allocation $X\in\cI$ that maximizes a given utility $u(X)$.
Note that $u$ and $\cI$ are given from $\cU$ and $\Gamma$. Formally, we call it the  $(\cU,\Gamma)$-packing problem, which corresponds with a $(\cU,\Gamma)$-market. 
\begin{theorem}\label{thm:check}
Fix a set of utility functions $\cU$ and a set of independence families $\Gamma$.
If the $(\cU,\Gamma)$-packing problem is solvable in polynomial time,
then the $\alpha$-stability of a given matching in a $(\cU,\Gamma)$-market can be checked in polynomial time for any $\alpha~(\ge 1)$.
If it is coNP-hard to determine whether a given solution to an instance of the $(\cU,\Gamma)$-packing problem is $\alpha$-approximate,
then it is coNP-hard to test whether a given matching is $\alpha$-stable in a $(\cU,\Gamma)$-market. 
\end{theorem}
\begin{proof}
We first prove the former part.
Let $(D,H,{\succ_D},u_H,\allowbreak\cI_H)$ be a $(\cU,\Gamma)$-market and let $\mu\subseteq D\times H$ be a matching.
Then $\mu$ is $\alpha$-stable if and only if
\begin{align*}
   \alpha\cdot u(\mu(h))\ge
   \max\bigl\{u(D') \bigm|
   D'\in\cI_h|D_h
   \bigr\}
\end{align*}
for all $h\in H$, where $D_h\coloneqq\{d\in D\mid h\succeq_d \mu(d)\}$ and $\cI_h|D_h$ is the restriction $\cI_h$ to $D_h$.
The right-hand side value is computed in polynomial time if 
the corresponding $(\cU,\Gamma)$-packing problem is solvable in polynomial time. 
Thus, the $\alpha$-stability is checked efficiently.

Next, we give a reduction to prove the latter part.
For an instance $(u,\cI)$ of the $(\cU,\Gamma)$-packing problem,
let us consider a $(\cU,\Gamma)$-market with doctors $D$ and one hospital $H\coloneqq\{h^*\}$.
Suppose that 
$h^*\succ_{d}\emptyset$ for all $d\in D$, 
$u_{h^*}\coloneqq u$, and 
$\cI_{h^*}\coloneqq\cI$.
We then reduce the $(\cU,\Gamma)$-market to $(D,H,{\succ_D},\allowbreak{}u_H,\cI_H)$. 
Moreover, matching $\mu$ is $\alpha$-stable if and only if
$u_{h^*}(\mu(h^*))$ is an $\alpha$-approximation of $\max\left\{u(D')\mid D'\in\cI_{h^*}\right\}$. 
Thus the claim holds.
\end{proof}

The theorem enables us to access the proficiency of packing problems.  
For example, since the $(\Uadd,\Gmat{2})$-packing problem (i.e., the weighted matroid intersection problem) is solvable in polynomial time~\cite{Edmonds1970}, one can efficiently check the $\alpha$-stability of a given matching in a $(\Uadd,\Gmat{2})$-market. 
Garey and Johnson~\cite{garey1979cai} provided the coNP-completeness of several packing problems. 
They notify us that checking the stability of a matching is coNP-complete in the corresponding markets, summarized in Table~\ref{tab:summary}\,(\subref{stab:stab}).

\section{Hardness of Computing a Stable Matching}\label{sec:hardness_existence}

In this section, we discuss the negative side of computing an $\alpha$-stable matching.
Kojima \emph{et al.}~\cite{kty:jet:2018} reveals that we can efficiently find a $1$-stable matching for 
any $(\Uadd,\Gmat{1})$-market and Kawase and Iwasaki~\cite{KI2017} proves the same for $(\Ucard,\Gvec{1,0})$-market. 
In general, the existence problems we consider belong to \StP{}, since \emph{yes}-instance can be verified by checking the stability of a guessed $1$-stable matching with the NP-oracle.

We can say that it is NP-hard to find (or determine the nonexistence of) an $\alpha$-stable matching in a $(\cU,\Gamma)$-market if it is NP-hard to compute an $\alpha$-approximate solution to a $(\cU,\Gamma)$-packing instance, by applying the similar argument in Theorem~\ref{thm:check}. 
Furthermore, we can conclude that the existence problem for the hard cases are all \StP-complete.
Note that the \StP-completeness for the $(\Uadd,\Gvec{1})$-markets has shown by Hamada \textit{et al.}~\cite{HISY2017}.
We prove the hardness for the other cases by reductions from the \EAtDM{} or \EASSP{}, which are \StP-complete~\cite{berman1997,mcloughlin1984}. 
{
\begin{description}
\item[\EAtDM] We are given three disjoint sets $X_1,X_2,X_3$ of the same cardinality, and two disjoint subsets $S^\exists,S^\forall\subseteq X_1\times X_2\times X_3$. Our task is to determine whether there exists $T^\exists\subseteq S^\exists$ so that $T^\exists\cup T^\forall$ is not a matching for any $T^\forall\subseteq S^\forall$.

\item[\EASSP] We are given two disjoint sets $S^\exists,S^\forall$ with weights $a\colon S^\exists\cup S^\forall\to \mathbb{Z}_{+}$, and an integer $q$.
Our task is to determine whether there exists $T^\exists\subseteq S^\exists$ so that $\sum_{e\in T^\exists\cup T^\forall}a(e)\ne q$ for any $T^\forall\subseteq S^\forall$.
\end{description}}
We show Theorems~\ref{thm:Sigma2Pa} and \ref{thm:Sigma2Pb} by reductions from \EAtDM{} and Theorem~\ref{thm:Sigma2Pc} by a reduction from \EASSP{}. 

\begin{theorem}\label{thm:Sigma2Pa}
It is \StP-hard to decide whether a given $(\Usub,\Gcap)$-market has a $1$-stable matching.
\end{theorem}
\begin{theorem}\label{thm:Sigma2Pb}
It is \StP-hard to decide whether a given $(\Ucard,\Gmat{3})$-market has a $1$-stable matching.
\end{theorem}
\begin{theorem}\label{thm:Sigma2Pc}
It is \StP-hard to decide whether a given $(\Ucard,\Gvec{2})$-market has a $1$-stable matching.
\end{theorem}
\begin{proof}
Here we only provide a proof for Theorem~\ref{thm:Sigma2Pa}.
Proofs for Theorems~\ref{thm:Sigma2Pb} and~\ref{thm:Sigma2Pc} can be obtained in similar ways (formal proofs are shown in Appendix~\ref{apx:hardness_existence}).

We give a reduction from \EAtDM{}.
Suppose that disjoint sets $S^\exists,S^\forall\subseteq X_1\times X_2\times X_3$ are given as an instance of \EAtDM.

We construct a $(\Usub,\Gcap)$-market that has a $1$-stable matching if and only if the given instance is a \emph{yes}-instance.
Consider a market $(D,H,\succ_D,u_H,\cI_H)$ with 
$D\coloneqq\{d_1,d_2,d_3,d_4\}\cup \{d^e\}_{e\in S^\exists\cup S^\forall}$ and 
$H\coloneqq \{h^*,h_1,h_2\}\cup \{h^e\}_{e\in S^\exists}$.
The doctors' preferences over the acceptable hospitals are given as:
\begin{multicols}{3}
\begin{itemize}
\item $d^e\colon h^e,h^*$ $(e\in S^\exists)$,
\item $d^e\colon h^*$ $(e\in S^\forall)$,
\item $d_1\colon h^*,h_1$,
\item $d_2\colon h_1$,
\item $d_3\colon h_2,h_1$,
\item $d_4\colon h_1,h_2$.
\end{itemize}
\end{multicols}
\noindent Here, and henceforth, preference lists are ordered from left to right in decreasing order of preference.
The feasibility constraint is the capacity constraint of 
rank 1 for $h_2$ and $h^e~(e\in S^\forall)$,
rank 2 for $h_1$, and 
rank $|X_1|~(=|X_2|=|X_3|)$ for $h^*$.
Suppose that $u_{h_1}$ and $u_{h_2}$ are the same as Example~\ref{ex:submo} (the utilities of unmatchable doctors are considered to be zero),
and $u_{h^e}$ is identically zero $(\forall e\in S^\exists)$.
In addition, for $X\subseteq D$, we define
\[u_{h^*}(X)\coloneqq \left|\textstyle\bigcup_{d^{(x_1,x_2,x_3)}\in X\setminus\{d_1,d_2,d_3,d_4\}}\{x_1,x_2,x_3\}\right|+2|X\cap\{d_1\}|.\]
This is a weighted-coverage function and hence submodular.

Consider the case when the instance is a \emph{yes}-instance.
We show that there exists a $1$-stable matching in this case.
Let $\hat{T}^\exists\subseteq S^\exists$ be a certificate of the instance. 
Without loss of generality, we may assume that $|\hat{T}^\exists|\le |X_1|-1$, $|\hat{T}^\exists|+|S^\forall|\ge |X_1|$, and $\hat{T}^\exists$ is a matching. 
Let 
\[\hat{T}^\forall\in\argmax\Bigl\{u_{h^*}\bigl(\{d^e\}_{e\in \hat{T}^\exists}\cup \{d^e\}_{e\in T^\forall}\bigr)\,\Bigm|\, |\hat{T}^\exists|+|T^\forall|=|X_1|,\ T^\forall\subseteq S^\forall\Bigr\}.\]
Then, $\hat{T}^\exists\cup \hat{T}^\forall$ is not a matching, and hence there exists $e^*\in\hat{T}^\forall$ such that 
\[u_{h^*}\bigl(\{d^e\}_{e\in \hat{T}^\exists}\cup \{d^e\}_{e\in \hat{T}^\forall}\bigr)
\le u_{h^*}\bigl(\{d^e\}_{e\in \hat{T}^\exists}\cup \{d^e\}_{e\in \hat{T}^\forall\setminus\{e^*\}}\bigr)+2.\]
Thus, the matching
\begin{align*}
\mu^*=
&\{(d^e,h^*)\}_{e\in\hat{T}^\exists}
\cup \{(d^e,h^*)\}_{e\in\hat{T}^\forall\setminus\{e^*\}}
\cup \{(d^e,h^e)\}_{e\in S^\exists\setminus\hat{T}^\exists}
\cup\{(d_1,h^*),(d_2,h_1),(d_3,h_2),(d_4,h_1)\}
\end{align*}
is $1$-stable.

Conversely, consider the case when the instance is a \emph{no}-instance.
We show that there exists no $1$-stable matching in this case.
Suppose to the contrary that $\mu$ is a $1$-stable matching.
Then, $\mu$ must contain $(d_1,h^*)$
since the submarket induced by $\{d_1,d_2,d_3,d_4\}$ and $\{h_1,h_2\}$ 
is equivalent to Example~\ref{ex:submo} (which has no $1$-stable matching).
Hence, $u_{h^*}(\mu(h^*))\le 3|X_1|-1$.
Let $\tilde{T}^\exists=\{e\in S^\exists\mid d^e\in \mu(h^*)\}$.
Since the instance is a \emph{no}-instance, there exists $\tilde{T}^\forall\subseteq S^\forall$ such that 
$|\tilde{T}^\exists|+|\tilde{T}^\forall|=|X_1|$ and $\tilde{T}^\exists\cup \tilde{T}^\forall$ is a matching.
Thus, we have $u_{h^*}\bigl(\{d^e\}_{e\in \tilde{T}^\exists}\cup\{d^e\}_{e\in \tilde{T}^\forall}\bigr)=3|X_1|$,
which implies that $\tilde{T}^\exists\cup\tilde{T}^\forall$ is a $1$-blocking coalition for $h^*$.
\end{proof}

Now the remaining cases to be treated are $(\Ucard,\Gmat{2})$- and $(\Uadd,\Gmat{2})$-markets.
For a $(\Uadd,\Gmat{2})$-market, although the $\alpha$-stability (especially $1$-stability) of a given matching can be checked in polynomial time by Theorem~\ref{thm:check}, the existence problem becomes NP-complete, even if utilities are restricted to cardinality. 
We prove the NP-hardness by a reduction from \DMP{}, which is NP-complete~\cite{frieze1983}.

\begin{description}
\item[\DMP] We are given two bipartite graphs, $(S,T;A_1)$ and $(S,T;A_2)$ with $|S|=|T|$, and our task is to determine whether perfect matchings $M_1\subseteq A_1$ and $M_2\subseteq A_2$ exist such that $M_1\cap M_2=\emptyset$.
\end{description}

\begin{theorem}\label{thm:2hard}
  It is NP-complete to decide whether a given $(\Ucard,\Gmat{2})$-market has a $1$-stable matching or not.
\end{theorem}
\begin{proof}
The problem is clearly in NP, since one can efficiently check the $1$-stability of a given matching in a $(\Ucard,\Gmat{2})$-market.

We give a reduction from \DMP{} to prove NP-hardness.
Let $(S,T;A_1)$ and $(S,T;A_2)$ be the bipartite graphs of a given disjoint matching instance and
let $A_1\cup A_2=\{a_1,\dots,a_\ell\}$.
Without loss of generality, we assume that $|A_1\cup A_2|\ge 2|S|$.

We construct a market that has a $1$-stable matching if and only if the given instance has disjoint perfect matchings.
Consider a market $(D,H,\succ_D,u_H,\cI_H)$ with
$4\ell$ doctors $D\coloneqq\bigcup_{k=1}^\ell \{d^k_1,d^k_2,d^k_3,d^k_4\}$ and
$2\ell+3$ hospitals $H\coloneqq\{h_1,h_2,h_3\}\cup\bigcup_{k=1}^{\ell}\{h^k_1,h^k_2\}$.
The doctors' preferences over the acceptable hospitals are given as:
\begin{multicols}{2}
\begin{itemize}
\item $d^{k}_1\colon h_1, h_2, h_3, h_{k}^1, h_{k}^2$ \ $(k\in [\ell])$,
\item $d^{k}_2\colon h^{k}_1, h^k_{2}$ \ $(k\in [\ell])$,
\item $d^{k}_3\colon h^{k}_2, h^k_{1}$ \ $(k\in [\ell])$,
\item $d^{k}_4\colon h^{k}_2, h^k_{1}$ \ $(k\in [\ell])$.
\end{itemize}
\end{multicols}
Suppose that each hospital has the cardinality utility.
We equate each doctor $d_1^k\in D$ with edge $a_k\in A_1\cup A_2$.
Then, the feasibility constraint for each hospital is defined:\\[4mm]
\begin{minipage}[t]{.6\linewidth}
\begin{itemize}
\item $\cI_{h_1}\coloneqq\{D'\subseteq A_1\mid \text{$D'$ is a matching in $(S,T;A_1)$}\}$,
\item $\cI_{h_2}\coloneqq\{D'\subseteq A_2\mid \text{$D'$ is a matching in $(S,T;A_2)$}\}$,
\item $\cI_{h_3}\coloneqq\{D'\subseteq A_1\cup A_2\mid |D'|\le |A_1\cup A_2|-2|S|\}$,
\end{itemize}
\end{minipage}%
\begin{minipage}[t]{.4\linewidth}
\begin{itemize}
\item $\cI_{h^k_1}\coloneqq 2^{\{d^k_1,d^k_3\}}\cup 2^{\{d^k_2,d^k_4\}}$ $(k\in [\ell])$,
\item $\cI_{h^k_2}\coloneqq 2^{\{d^k_1,d^k_4\}}\cup 2^{\{d^k_2,d^k_3\}}$ $(k\in [\ell])$.
\end{itemize}
\end{minipage}\\[4mm]
Note that each feasible family can be represented as an intersection of two matroids.

Consider the case when the instance has disjoint perfect matchings.
Let $M_1\subseteq A_1$ and $M_2\subseteq A_2$ be the matchings.
Then
\begin{align*}
  \mu=&\{(d,h_1)\}_{d\in M_1}
       \cup \{(d,h_2)\}_{d\in M_2}
       \cup \{(d,h_3)\}_{d\in (A_1\cup A_2)\setminus (M_1\cup M_2)}
       \cup \textstyle\bigcup_{k=1}^\ell \{(d_2^k,h_1^k),(d_3^k,h_2^k),(d_4^k,h_1^k)\}
\end{align*}
is a $1$-stable matching.

Conversely, consider the case when the instance has no disjoint perfect matchings.
Let $\mu$ be a feasible matching.
Then, there exists a doctor $d_1^k$ such that $\mu(d_1^k)\not\in\{h_1,h_2,h_3\}$.
In this case, $\mu$ is not $1$-stable 
since doctors $\{d_1^k,d_2^k,d_3^k,d_4^k\}$ and hospitals $\{h_1^k,h_2^k\}$ form the same market as in Example~\ref{ex:nonexistence}.
\end{proof}

\section{Approximability of Stable Matchings}
To deal with the nonexistence or the hardness of stable matchings, we focus on an approximately stable matching where stability may be violated to some extent. This section pays an attention to an online version of packing problems, i.e., \emph{online packing problems} (with cancellation) and incorporates the proficiency into a variant of generalized deferred acceptance (GDA) algorithm~\cite{Hatfield:AER:2005} in such a manner that choice functions of hospitals are replaced with an online packing algorithm. We establish a framework so that the bounds of the algorithms become consistent with how much stability is violated. Note that Kawase and Iwasaki~\cite{KI2018} apply a similar idea for budget constraints, that is, $1$-dimensional knapsack with $\epsilon$-slack constraints. 

In what follows, we consider algorithms that take a market as input and yield an approximately stable matching as output. 
An \emph{algorithm} is called \emph{$\alpha$-stable} if it always produces an $\alpha$-stable matching for a certain $\alpha$. 

\subsection{Online Packing Problem}
Let us briefly introduce an online packing problem, 
which is a generalization of several online problems such as an online removable knapsack problem~\cite{iwama2002rok}.
Its instance consists of 
a set of elements $D=\{d_1,\dots,d_n\}$, 
a utility function $u\colon 2^D\to \mathbb{R}_+$, and 
a feasibility family $\cI\subseteq 2^D$. 
We assume that $u$ is monotone and $\cI$ is an independence family. 
Elements in $D$ are given to an online algorithm $\ALG$ one by one in an unknown order. 
When an element is presented, the algorithm must accept or reject it immediately 
without knowledge about the ordering of future elements.
Although accepted elements can be canceled, the elements that are once rejected (or canceled) can never be recovered.
The set of selected elements must be feasible in each round.
Suppose that elements are given according to order $\sigma$, which is a bijection from $[n]$ to $D$.
We denote by $\ALG(\sigma(1),\dots,\sigma(i))$ the set of selected elements at the end of the $i$th round, in which $\sigma(i)\in D$ arrives.
We denote it as $\ALG(\sigma,i)$ for brevity.
Then we have $\ALG(\sigma,0)=\emptyset$, and 
$\ALG(\sigma,i-1)\cup\{\sigma(i)\}\supseteq \ALG(\sigma,i)\in\cI$
holds for every $i\in[n]$.
In addition, for $i\in[n]$ and the two orders of elements $\sigma$ and $\tau$,
equality $\ALG(\sigma,i)=\ALG(\tau,i)$ holds if $\sigma(j)=\tau(j)$ for any $j\in[i]$.
Our task is to maximize value $u(\ALG(\sigma,i))$ for unknown order $\sigma$ and $i~(\in[n])$.

The performance of an online algorithm is measured by the \emph{competitive ratio}.
Denote 
\[\OPT(\sigma,i)\in\argmax\bigl\{u(S)  \bigm| S\in\cI|\{\sigma(1),\dots,\sigma(i)\}\bigr\}.\]
Online algorithm $\ALG$ is called \emph{$\alpha$-competitive} ($\alpha\ge 1$) if
\begin{align*}
\alpha\cdot u(\ALG(\sigma,i))\ge u(\OPT(\sigma,i))
\end{align*}
for any $\sigma$ and $i$.
An online $(\cU,\Gamma)$-packing problem is called $\alpha$-competitive if an $\alpha$-competitive algorithm exists for any instance of it.

\subsection{Generalized Deferred Acceptance Algorithm}
We use a modified version of the generalized DA algorithm, which is formally described in Algorithm~\ref{alg:GDA}.
In GDA, each doctor is initialized to be unmatched. 
Then an unmatched doctor makes a proposal to her most preferred hospital $h$ that has not rejected her yet.
Let $\bm{a}^{h}$ be the ordered list of proposed doctors to $h$.
Then it chooses a set of doctors according to the output of online algorithm $\ALG_{h}(\bm{a}^{h})$.
The proposal procedure continues as long as an unmatched doctor has a non-rejected acceptable hospital.

\begin{algorithm}[t]
  \SetKwInOut{Input}{input}\Input{\small $D,H,(\succ_d)_{d\in D},{(\ALG_h)}_{h\in H}$\quad\textbf{output:} matching $\mu$}
  \caption{Generalized DA algorithm}\label{alg:GDA}
  $\mu\ot\emptyset$, $R_d\ot H$ ($\forall d\in D$)\;
  $L\ot \{d\in D\mid \max_{\succ_d} (R_d\cup\{\emptyset\})\ne\emptyset\}$\;
  $\bm{a}^h\ot ()$ for all $h\in H$\;
  \While{$L\ne\emptyset$}{
    pick $d\in L$ arbitrarily and let $h\ot \max_{\succ_d} R_d$\;\label{line:pick_contract}
    append $d$ to the end of $\bm{a}^{h}$\;
    $\mu\ot \{(d',h')\in\mu\mid h'\ne h\}\cup\ALG_h(\bm{a}^h)$\;
    $R_d \ot R_d\setminus\{h\}$\;
    $L\ot\{d\in D\mid \max_{\succ_d} (R_d\cup\{\mu(d)\})\ne\mu(d)\}$\;
  }
  \Return $\mu$\;
\end{algorithm}

The next theorem guarantees that if $\ALG_h$ is $\alpha$-competitive for each $h\in H$, then Algorithm~\ref{alg:GDA} is $\alpha$-stable.%
\footnote{Although the output of Algorithm~\ref{alg:GDA} depends on the order of doctors selected in Line~\ref{line:pick_contract}, this claim holds regardless of the order.}

\begin{theorem}\label{thm:GDA}
  If the online $(\cU,\Gamma)$-packing problem is $\alpha$-competitive,
  then an $\alpha$-stable algorithm exists for the $(\cU,\Gamma)$-markets. 
\end{theorem}
\begin{proof}
We prove that the output $\mu$ is an $\alpha$-stable matching by contradiction.
Suppose that $D'\subseteq D$ is an $\alpha$-blocking coalition for $h$.
Then, we have $u_h(D')>\alpha\cdot u_h(\mu(h))=\alpha\cdot u_h(\ALG_h(\bm{a}^h))$ and $h\succeq_d \mu(d)$ for all $d\in D'$.
By the definition of the algorithm, $h\succeq_d \mu(d)$ implies that $d$ is in $\bm{a}^h$.
Hence, we have $\alpha\cdot u_h(\ALG_h(\bm{a}^h))\ge u_h(\OPT_h(\bm{a}^h))\ge u_h(D')$
since $\ALG$ is $\alpha$-competitive.
This is a contradiction.
\end{proof}
This theorem assures that if there exists an online packing algorithm in a setting, we can construct a stable algorithm in the corresponding market with it.

Let us first apply a greedy algorithm to a matching problem for $k$-matroid constraints: 
Start from the empty solution and add an element to the current solution if and only if its addition preserves feasibility.
If the utilities are cardinality, it is a $k$-competitive algorithm~\cite{jenkyns1976teo,korte1978aao}.
By Theorem~\ref{thm:GDA}, we obtain the following corollary.
\begin{corollary}\label{thm:cardmat}
  There exists a $k$-stable algorithm for the $(\Ucard,\Gmat{k})$-markets.
\end{corollary}

Hence, Algorithm~\ref{alg:kmatgreedy} is $k$-competitive for the online $(\Ucard,\Gmat{k})$-packing problem.\footnote{This claim can be generalized to the case when the set of independence families is the \emph{$k$-system}, which is an extension of $k$-matroid intersection.
An independence system $(D,\cI)$ is called a $k$-system if for all $S\subseteq D$, the ratio of the cardinality of the largest to the smallest maximal independent subset of $S$ is at most $k$.}

\begin{algorithm}
  \SetKwInOut{Input}{input}\Input{\small $\sigma(1),\dots,\sigma(i)$\quad\textbf{output:} $\ALG(\sigma(1),\dots,\sigma(i))$}
  \caption{}\label{alg:kmatgreedy}
  \lIf{$i=0$}{\Return $\emptyset$}
  let $Y\ot \ALG(\sigma(1),\dots,\sigma(i-1))\cup\{\sigma(i)\}$\;
  \lIf{$Y\in\cI$}{\Return $Y$}
  \lElse{\Return $\ALG(\sigma(1),\dots,\sigma(i-1))$}
\end{algorithm}

However, we require more sophisticated packing algorithms in the online setting to handle the additive or submodular utility case.  

For the additive case, a $(\sqrt{k}+\sqrt{k-1})^2$-competitive algorithm was given by Ashwinkumar~\cite{ashwinkumar2011ami}.
For the submodular case, a $4k$-competitive algorithm was given by Chakrabarti and Kale~\cite{CK2015}.
Thus, we obtain the following corollaries:
\begin{corollary}\label{thm:addmat}
  There exists a $(\sqrt{k}+\sqrt{k-1})^2$-stable algorithm for the $(\Uadd,\Gmat{k})$-markets.
\end{corollary}
\begin{corollary}\label{thm:submat}
  There exists a $4k$-stable matching algorithm for the $(\Usub,\Gmat{k})$-markets.
\end{corollary}

Next, let us consider $\rho$-dimensional knapsack constraints with $\epsilon$-slack.
To the best of our knowledge, no suitable online packing algorithm has been proposed for the cardinality or additive utility case. 
We develop a simple greedy algorithm, which is formally given as Algorithm~\ref{alg:vecgreedy}.
Intuitively, the algorithm chooses doctors according to decreasing order of utility per largest size $u(d)/\max_{i\in[\rho]}w(d,i)$.
\begin{algorithm}[t]
  \SetKwInOut{Input}{input}\Input{\small $\sigma(1),\dots,\sigma(i)$\quad\textbf{output:} $\ALG(\sigma(1),\dots,\sigma(i))$}
  \caption{}\label{alg:vecgreedy}
  \lIf{$i=0$}{\Return $\emptyset$}
  let $Y\ot \ALG(\sigma(1),\dots,\sigma(i-1))\cup\{\sigma(i)\}$\;
  \While{$\sum_{d\in Y}\max_{i\in[\rho]}w(d,i)>1$}{
    $Y\ot Y\setminus\{a\}$ with $a\in\argmin_{d\in Y}\frac{u(d)}{\max_{i\in[\rho]}w(d,i)}$\;
  }
  \Return $Y$\;
\end{algorithm}
It is not difficult to see that the algorithm is $\rho$-competitive for the cardinality utility case
and $\rho/\epsilon$-competitive for the additive utility case.
We provide the formal proof in Appendix~\ref{sec:vecgreedy}.
\begin{theorem}\label{thm:vecgreedy}
Algorithm~\ref{alg:vecgreedy} is $\rho$-competitive for the online $(\Ucard,\Gvec{\rho,\epsilon})$-packing problem
and $\rho/\epsilon$-competitive for the online $(\Uadd,\Gvec{\rho,\epsilon})$-packing problem.
\end{theorem}

Accordingly, in conjunction with Theorem~\ref{thm:GDA}, we obtain the following corollaries.
\begin{corollary}\label{thm:cardknap}
  There exists a $\rho$-stable algorithm for the $(\Ucard,\Gvec{\rho,\epsilon})$-markets.
\end{corollary}
\begin{corollary}\label{thm:addknap}
  There exists a $\frac{\rho}{\epsilon}$-stable algorithm for the $(\Uadd,\Gvec{\rho,\epsilon})$-markets.
\end{corollary}

Unfortunately, it is not easy to generalize the greedy algorithm to the submodular utility case.
However, we can borrow an $O(\rho/\epsilon^2)$-competitive algorithm for the online $(\Usub,\Gvec{\rho,\epsilon})$-packing problem~\cite{CJTW2017} and provide the following corollary.
\begin{corollary}\label{thm:subknap}
  There exists an $O(\rho/\epsilon^2)$-stable algorithm for the $(\Usub,\Gvec{\rho,\epsilon})$-markets.
\end{corollary}

\section{Inapproximability of Stable Matchings}\label{sec:inapprox}
In this section, we show some inapproximability results, i.e., lower bounds. 
As we saw in Example~\ref{ex:nonexistence}, 
if utilities are cardinality and constraints are given from $\Gmat{2}\cap\Gvec{2,\epsilon}$, 
for any $\epsilon\in [0,1/2)$, 
there exists no $\alpha$-stable algorithm whose approximation ratio is better than two ($\alpha<2$). 
Also, Kawase and Iwasaki~\cite{KI2018} derive the lower bound for the case of additive utilities and $1$-dimensional knapsack with $\epsilon$-slack constraints. 
There exists a market that has no $\alpha$-stable matching with $\alpha<1/\epsilon$ if $\epsilon\in[0,1/2)$.

To fill up the remaining shown in Table~\ref{tab:summary}\,(\subref{stab:apx}), we now provide a basis for deriving lower bounds.
The next theorem reveals that given an $\alpha$-competitive algorithm for an online $(\cU,\Gamma)$-packing problem, the bound is equivalent to the extent to which stability is violated in the corresponding market. 

To this end, we have to restrict the input orderings. One might think that this makes the online problem too easy. 
However, without the restriction, we cannot directly use the lower bounds of the competitive ratio for the original online problem. 
Roughly speaking, we partition elements $D$ to $D_1,\dots,D_s$ and only consider input sequences such that, if an online algorithm rejects an element in $D_t$ for each $t\in [s]$, then a new element in $D_t$ is given to the algorithm.
In addition, we require that 
the $(\cU,\Gamma)$-markets allow hospitals to have any additive utilities and the rank one capacity.
\begin{theorem}\label{thm:general_lower}
Suppose that $\cU$ is a set of utilities and $\Gamma$ is a set of independence families
such that $\Uadd\subseteq \cU$ and $\Gcap^{(1)}\subseteq\Gamma$. 
Let $(\hat{D},\hat{u},\hat{\cI})$ be a $(\cU,\Gamma)$-packing instance, let
$D_1,\dots,D_s$ be a partition of $\hat{D}$ with $D_t=\{d^t_{1},\dots,d^t_{r_t}\}$ for each $t\in[s]$ (where $r_t=|D_t|$), and let
$q_1,\dots,q_s$ be positive integers 
such that $q_t\le r_t$ for each $t\in[s]$.
We define 
\[\cl(D')\coloneqq\bigcup_{t\in[s]}\Big\{d^t_{i}\in D_t\Bigm| |D'\cap D_t|<q_t~\text{or}~i\le \!\max_{d^t_{j}\in D'\cap D_t}j\Big\}\]
and $\mathcal{A}\coloneqq\{D'\subseteq \hat{D}\mid |D'\cap D_t|\le q_t~(\forall t\in[s])\}$. 
Then, there exists no $\alpha$-stable algorithm for the $(\cU,\Gamma)$-markets if 
there exists no $\alpha$-competitive algorithm for the online packing problem with restricted input sequences, i.e.,
\begin{align}
\alpha\cdot \hat{u}(D')< \max\{\hat{u}(D'')\mid D''\in\hat{\cI}|{\cl(D')}\}, \label{eq:general_lower}
\end{align}
for any $D'\in\hat{\cI}\cap\mathcal{A}$. 
\end{theorem}
\begin{proof}
Suppose that \eqref{eq:general_lower} holds for any $D'\in\cI\cap\mathcal{A}$.
We set hospitals as \(H\coloneqq\{h^*\}\cup\textstyle\bigcup_{t\in [s]}\{h_{q_t+1}^{t},\dots,h_{r_t}^{t}\}\).
Each doctor's preference over her acceptable hospitals is defined as:
\begin{itemize}
\item $\succ_{d_i^t}\colon h^*, h_{q_t+1}^t,\dots,h_{r_t}^t$ \ $(t\in[s],~i\in[q_t])$,
\item $\succ_{d_i^t}\colon h_{i}^t, h^*, h_{i+1}^t,\dots,h_{r_t}^t$ \ $(t\in[s],~i\in[r_t]\setminus[q_t])$.
\end{itemize}
Suppose that $u_{h^*}(D')=\hat{u}(D')$ $(\forall D'\subseteq D)$ and 
each hospital $h_i^t$ has an additive utility function such that
$u_{h_i^t}(d_j^t)=(\alpha+1)^{-j}$ \ $(t\in[s],~i\in [r_t]\setminus [q_t],~j\in [r_t])$.
The feasibility constraint of each hospital is defined as:
$\cI_{h^*}=\hat{\cI}$ and $\cI_{h_i^t}=\{D'\subseteq D\mid |D'|\le 1\}$ $(t\in[s],~i\in[r_t])$.
Note that $(D,H,\succ_D,u_H,\cI_H)$ is a $(\cU,\Gamma)$-market.

In what follows, we claim that no $\alpha$-stable matching exists in market $(D,H,\succ_D,u_H,\cI_H)$
if \eqref{eq:general_lower} holds.
Suppose, contrary to our claim, that $\mu$ is an $\alpha$-stable matching.

We show that $h^*\succeq_d\mu(d)$ if $d\in\cl(\mu(h^*))$.
Fix $t\in[s]$ and let $D_t\setminus \mu(h^*)=\{d^t_{\sigma(1)},\dots,d^t_{\sigma(\ell)}\}$
with $\ell=|D_t\setminus \mu(h^*)|$ and $\sigma(1)<\dots<\sigma(\ell)$.
It is worth mentioning that $\sigma(i)\le q_t+i$ for all $i\in[\ell]$.
By $u_{h_{q_t+1}^t}(d^t_{\sigma(1)})>\alpha\cdot u_{h_{q_t+1}^t}$ and $\sigma(i)\le q_t+1$, we have $\mu(h_{q_t+1}^t)=d^t_{\sigma(1)}$.
Similarly, by induction we conclude that $\mu(h_{q_t+i}^t)=d^t_{\sigma(i)}$ for all $i\in[\ell]$.
Here, we have $\mu(d^t_i)\succ_{d^t_i}h^*$ if and only if $i=\sigma(i-q_t)$.
Hence, we have $h^*\succeq_{d^t_i}\mu(d^t_i)$ if $\ell>r_t-q_t$ (i.e., $|\mu(h^*)\cap D_t|<q_t$) or $d^t_j\in\mu(h^*)$ for some $j\ge i$ (i.e., $i\le\max_{d^t_{j}\in D'\cap D_t}j\}$).

Therefore, $D^*\in\argmax\{u_{h^*}(D'')\mid D''\in\hat{\cI}|{\cl(D')}\}$ is an $\alpha$-blocking pair for $\mu$
by $D^*\in\cI_{h^*}$, $h^*\succeq_d \mu(d)$ ($\forall d\in D^*$), and $u_{h^*}(D^*)>\alpha\cdot u_{h^*}(\mu(h^*))$,
which contradicts the $\alpha$-stability of $\mu$.
\end{proof}

Theorem~\ref{thm:general_lower} gives us a general lower bound in \eqref{eq:general_lower} and it enable us to derive two novel lower bounds for $k$-matroid intersection and $\rho$-dimensional knapsack constraints.
\begin{theorem}\label{thm:lower_addmat}
 For any integer $k\ge 2$ and any real $(1\le)~\alpha<k$,
 there exists no $\alpha$-stable algorithm for the $(\Uadd,\Gmat{k})$-markets.
\end{theorem}
\begin{proof}
   Let $D=D_1\cup D_2$ with $D_1\coloneqq \{d_1^1,\dots,d_k^1\}$ and $D_2\coloneqq \{d_1^2,\dots,d_k^2\}$.
   In addition, let $q_1=q_2=1$, $\cI\coloneqq 2^{D_1}\cup 2^{D_2}~(\in \Gmat{k})$, and $u(D')=|D'|$ for $D'\subseteq D$.

   For set $\mathcal{A}$ in Theorem~\ref{thm:general_lower},
   we have $\cI\cap\mathcal{A}=\{\emptyset\}\cup\{d_i^1\mid i\in[k]\}\cup\{d_i^2\mid i\in[k]\}$.
   Thus, for any $D'\in\cI\cap\mathcal{A}$, we have $u(D')=|D'|\le 1$ and $\max\{u(D'')\mid D''\in\cI|{\cl(D')}\}=k$.
   Hence, by Theorem~\ref{thm:general_lower}, there exists a $(\Uadd,\Gmat{k})$-market without $\alpha$-stable matchings.
\end{proof}

\begin{theorem}\label{thm:lower_addknap}
  For any positive integer $\rho$ and any positive real $\epsilon<1/2$,
  there exists no $\frac{\rho}{2\epsilon}$-stable algorithm for the $(\Uadd,\Gvec{\rho,\epsilon})$-markets.
\end{theorem}
\begin{proof}
   Let $r\coloneqq \lceil 1/\epsilon\rceil-1$.
   We remark that $\epsilon<1/r<2\epsilon$.
   In addition, let $m$ be an integer such that $(r\rho)^{1-1/m}>\frac{\rho}{2\epsilon}$.
   Consider a $(\Uadd,\Gvec{\rho,\epsilon})$-packing problem with $mr\rho+1$ doctors
   $D=D_1\cup D_2$ where $D_1=\{d_0\}$ and $D_2=\{d_1\dots,d_{mr\rho}\}$.
   We define $q_1=q_2=1$.

   Let us partition $D^2$ into $D^{a,b}\coloneqq \{d_{t+r(a-1)+r\rho (b-1)}\mid t\in[r]\}$ ($a\in[\rho]$ and $b\in[m]$).
   Let $u$ be an additive utility function such that $u(d_0)=r\rho$ and $u(d_i)=(r\rho)^{b/m}$ for $d_i\in D^{a,b}$ with $a\in[\rho]$ and $b\in[m]$.
   The weights for dimension $\ell\in[\rho]$ are defined as:
   \begin{itemize}
   \item $w(d_0,\ell)=1-\epsilon$,
   \item $w(d_i,\ell)=1/r$ \ $(b\in[m];~d_i\in D^{\ell,b}$), and
   \item $w(d_i,\ell)=0$ \ $(a\in[\rho]\setminus\{\ell\};~b\in [m];~d_i\in D^{a,b})$.
   \end{itemize}

   For set $\mathcal{A}$ in Theorem~\ref{thm:general_lower},
   we have $\cI\cap\mathcal{A}=\{D'\subseteq D\mid |D'|\le 1\}$.
   We show that \eqref{eq:general_lower} holds for any $D'\in\cI\cap\mathcal{A}$.
   If $D'=\emptyset$, we have $\frac{\rho}{2\epsilon}\cdot u(D')=0<r\rho=\max\{u(D'')\mid D''\in\cI|{\cl(D')}\}$.
   Hence, we can assume that $D'=\{d_i\}$.
   If $i=0$, we have
   \[\frac{\max\{u(D'')\mid D''\in\cI|{\cl(D')}\}}{u(D')}
   =\frac{u(D^{1,m}\cup\dots\cup D^{\rho,m})}{u(\{d_0\})}
   =\frac{r\cdot\rho\cdot r\rho}{r\rho}=r\rho>\frac{\rho}{2\epsilon}.\]
   If $1\le i\le r\rho$ (i.e., $d_i\in D^{a,1}$ for some $a\in[\rho]$), we have
   \[\frac{\max\{u(D'')\mid D''\in\cI|{\cl(D')}\}}{u(D')}
   \ge \frac{u(\{d_0\})}{u(\{d_i\})}
   =\frac{r\rho}{(r\rho)^{\frac{1}{m}}}>\frac{\rho}{2\epsilon}.\]
   Finally, if $r\rho<i\le mr\rho$, let $q^*\coloneqq \lceil i^*/(r\rho)\rceil$ (i.e., $d_i\in D^{a,q^*}$ for some $a\in[\rho]$), and then we have
   \begin{align*}
   \frac{\max\{u(D'')\mid D''\in\cI|{\cl(D')}\}}{u(D')}
   &\ge \frac{u(D^{1,q^*-1}\cup\dots\cup D^{\rho,q^*-1})}{u(\{d_i\})}\\
   &\ge \frac{r\cdot\rho\cdot (r\rho)^{(q^*-1)/m}}{(r\rho)^{q^*/m}}
   =(r\rho)^{1-1/m}
   >\frac{\rho}{2\epsilon}.
   \end{align*}

   Thus, by Theorem~\ref{thm:general_lower}, there exists a $(\Uadd,\Gvec{\rho,\epsilon})$-market without $\frac{\rho}{2\epsilon}$-stable matchings.
\end{proof}

\section*{Acknowledgments}
This work was supported by JST ACT-I Grant Number JPMJPR17U7, and by JSPS KAKENHI Grant Numbers 16K16005, 16KK0003, and 17H01787.

\bibliographystyle{abbrv}
\bibliography{approxmatching}

\newpage
\appendix
\section{Proof of Nonexistence of Stable Matchings}
In this section, we formally prove that 
there exist no stable matchings in the markets described in Examples~\ref{ex:nonexistence} and~\ref{ex:submo}.

\begin{proposition}\label{prop:nonexistence}
The market described in Example~\ref{ex:nonexistence} has no $\alpha$-stable matching for any $1\le \alpha<2$.
\end{proposition}
\begin{proof}
  Recall that we consider the following market.
  There are four doctors $D\coloneqq\{d_1,d_2,d_3,d_4\}$ and two hospitals $H\coloneqq\{h_1,h_2\}$. 
  The preferences of the doctors are $h_1\succ_{d_i}h_2\succ_{d_i}\emptyset$ for $i=1,2$ and $h_2\succ_{d_i}h_1\succ_{d_i}\emptyset$ for $i=3,4$.
  Suppose that each hospital has a cardinality utility.
  The feasibility families are $\cI_{h_1}\coloneqq 2^{\{d_1,d_3\}}\cup 2^{\{d_2,d_4\}}$ and $\cI_{h_2}\coloneqq 2^{\{d_1,d_4\}}\cup 2^{\{d_2,d_3\}}$.

  We prove by contradiction that, for $1\le \alpha<2$, no $\alpha$-stable matching exists in this market.
  Suppose that $\mu$ is an $\alpha$-stable matching.

  \noindent\textbf{Case 1:} $|\mu(h_1)|=0$. Here $\{d_1\}$ is an $\alpha$-blocking coalition for $h_1$.
  
  \noindent\textbf{Case 2:} $|\mu(h_1)|=2$. $|\mu(h_2)|\le 1$.
  Thus, $\{d_2,d_3\}$ and $\{d_1,d_4\}$ are $\alpha$-blocking coalitions for a case when $\mu(h_1)=\{d_1,d_3\}$ and $\mu(h_1)=\{d_2,d_4\}$, respectively.

  \noindent\textbf{Case 3:} $|\mu(h_1)|=1$.
  If $\mu(h_1)=\{d_1\}$, then $(d_3,h_2)\in \mu$ since otherwise $\{d_1,d_3\}$ is an $\alpha$-blocking coalition for $h_1$.
  Thus, $(d_4,h_2)\not\in \mu$ and hence $\{d_2,d_4\}$ is an $\alpha$-blocking coalition for $h_1$.
  Similarly, $\mu$ cannot be $\alpha$-stable when $\mu(h_1)=\{d_2\}$.
  In addition, $\{d_1,d_3\}$ and $\{d_2,d_4\}$ are $\alpha$-blocking coalitions for $h_1$ when $\mu(h_1)=\{d_3\}$ and $\mu(h_1)=\{d_4\}$, respectively.
  
  Therefore, for $1\le \alpha<2$, no $\alpha$-stable matching exists in this market.
\end{proof}

\begin{proposition}\label{prop:submo}
The market described in Example~\ref{ex:submo} has no $1.28$-stable matching.
\end{proposition}
\begin{proof}
  Recall that we consider the following market.
  There are four doctors $D\coloneqq\{d_1,d_2,d_3,d_4\}$ and two hospitals $H\coloneqq\{h_1,h_2\}$. 
  The preferences of the doctors are 
  $h_1\succ_{d_i}\emptyset\succ_{d_i}h_2$ for $i=1,2$,
  $h_2\succ_{d_3}h_1\succ_{d_3}\emptyset$, and
  $h_1\succ_{d_4}h_2\succ_{d_4}\emptyset$.
  Suppose the $\cI_{h_1}$ and $\cI_{h_2}$ are capacity constraints of rank $2$ and rank $1$, respectively.
  Let $u_{h_1}$ be a submodular utility such that
  \begin{align*}
    u_{h_1}(D')\coloneqq \left|\bigcup_{d_i\in D'}A_i\right|
  \end{align*}
  where $A_1=\{a_{1,1},a_{1,2},a_{1,3}\}$, $A_2=\{a_{2,1},a_{2,2},a_{2,3}\}$, $A_3=\{a_{3,1},a_{3,2},a_{3,3}\}$, and $A_4=\{a_{1,1},a_{1,2},a_{2,1},a_{2,2}\}$.
  In addition, let $u_{h_2}$ be an additive utility such that $u_{h_2}(d_3)=1$ and $u_{h_2}(d_4)=2$.

  We prove, by contradiction, that there exists no $1.28$-stable matching in this market (where $1.28<(1+\sqrt{17})/4$).
  Suppose that $\mu$ is a $1.28$-stable matching.
  
  \noindent\textbf{Case 1:} $(d_4,h_1)\in\mu(h_1)$. 
  In this case, $\mu(h_2)=\{d_3\}$ since otherwise $\{d_3\}$ is a $1.28$-blocking coalition for $h_2$.
  Thus, $\mu(h_1)$ is $\{d_4\}$, $\{d_1,d_4\}$, or $\{d_2,d_4\}$.
  Hence, $u_{h_1}(\mu(h_1))\le 7+\sqrt{17}$ and $\{d_1,d_2\}$ is a $1.28$-blocking coalition for $h_1$ by $u_{h_1}(\{d_1,d_2\})=6+2\sqrt{17}=(7+\sqrt{17})(1+\sqrt{17})/4$.
  
  \noindent\textbf{Case 2:} $(d_4,h_1)\not\in\mu(h_1)$. 
  In this case, $u_{h_1}(\mu(h_1))\le 6+2\sqrt{17}$ and hence $\{d_3,d_4\}$ is a $1.28$-blocking coalition for $h_1$ by $u_{h_1}(\{d_3,d_4\})=10+2\sqrt{17}=(6+2\sqrt{17})(1+\sqrt{17})/4$.
\end{proof}

\section{Proofs of \StP-hardness}\label{apx:hardness_existence}

\newtheorem*{thm:Sigma2Pb}{Theorem~\ref{thm:Sigma2Pb}}
\begin{thm:Sigma2Pb}
It is \StP-hard to decide whether a given $(\Ucard,\Gmat{3})$-market has a $1$-stable matching.
\end{thm:Sigma2Pb}
\begin{proof}
We give a reduction from \EAtDM{}.
Suppose that disjoint sets $S^\exists,S^\forall\subseteq X_1\times X_2\times X_3$ are given as an instance of \EAtDM.
Without loss of generality, we may assume that $S^\exists$ does not contain conflicting elements
(i.e., $|\{x_1,x_2,x_3\}\cap\{x'_1,x'_2,x'_3\}|=0$
for any distinct $(x_1,x_2,x_3),(x'_1,x'_2,x'_3)\in S^\exists$) since otherwise we can easily certificate that the instance is a \emph{yes}-instance.

We construct a $(\Ucard,\Gmat{3})$-market that has a $1$-stable matching if and only if the given instance is a \emph{yes}-instance.
Consider a market $(D,H,\succ_D,u_H,\cI_H)$ with 
$D\coloneqq\{d_1,d_2,d_3,d_4,d^*\}\cup \{d^e\}_{e\in S^\exists\cup S^\forall}\cup\{\bar{d}^e\}_{e\in S^\exists}$ and 
$H\coloneqq \{h^*,h_1,h_2\}\cup \{h^e\}_{e\in S^\exists}$.
The doctors' preference over the acceptable hospitals are given as:
{
\setlength{\columnsep}{0pt}
\begin{multicols}{4}
\begin{itemize}
\item $d^e\colon h^e,h^*$ $(e\in S^\exists)$,
\item $\bar{d}^e\colon h^e$ $(e\in S^\exists)$,
\item $d^e\colon h^*$ $(e\in S^\forall)$,
\item $d^*\colon h^*$,
\item $d_1\colon h^*,h_1,h_2$,
\item $d_2\colon h_1,h_2$,
\item $d_3\colon h_2,h_1$,
\item $d_4\colon h_2,h_1$.
\end{itemize}
\end{multicols}}
The feasibility constraints are
$\cI_{h^e}\coloneqq \Gcap^{(1)}$ $(e\in S^\exists)$,
$\cI_{h_1}\coloneqq 2^{\{d_1,d_3\}}\cup 2^{\{d_2,d_4\}}$,
$\cI_{h_2}\coloneqq 2^{\{d_1,d_4\}}\cup 2^{\{d_2,d_3\}}$, and 
$\cI_{h^*}\coloneqq \cI^{(1)}\cap\cI^{(2)}\cap \cI^{(3)}$ where
\begin{align*}
\cI^{(1)}&\coloneqq 
\left\{T\subseteq D'
~\middle|\!
\begin{array}{l}
x_1\ne y_1\ \ (\substack{\forall d^x,d^y\in T\setminus\{d_1,d^*\}\\\text{such that }x\ne y})\\[2mm]
\end{array}\!\!\!\!
\right\},\\
\cI^{(2)}&\coloneqq 
\left\{T\subseteq D'
~\middle|\!
\begin{array}{l}
x_2\ne y_2\ \ (\substack{\forall d^x,d^y\in T\setminus\{d_1,d^*\}\\\text{such that }x\ne y})\\[2mm]
|T\setminus\{d^*\}|\le |X_1|
\end{array}\!\!\!\!
\right\},\\
\cI^{(3)}&\coloneqq 
\left\{T\subseteq D'
~\middle|\!
\begin{array}{l}
x_3\ne y_3\ \ (\substack{\forall d^x,d^y\in T\setminus\{d_1,d^*\}\\\text{such that }x\ne y})\\[2mm]
|T\cap \{d_1,d^*\}|\le 1
\end{array}\!\!\!\!
\right\}
\end{align*}
with $D'\coloneqq \{d_1,d^*\}\cup\{d^e\}_{e\in S^\exists\cup S^\forall}$.
Note that each constraint can be represented by an intersection of at most three matroids.
Moreover, suppose that the utility of each hospital is cardinality.

Consider the case when the instance is a \emph{yes}-instance.
We show that there exists a $1$-stable matching in this case.
Let $\hat{T}^\exists\subseteq S^\exists$ be a certificate of the instance, i.e.,
$|\hat{T}^\exists\cup T^\forall|<|X_1|$ for all $T^\forall\subseteq S^\forall$ such that $\hat{T}^\exists\cup T^\forall$ is a matching.
Let $\hat{T}^\forall\in\argmax\{|T^\forall|\mid \hat{T}^\exists\cup T^\forall\in\cI_{h^*},\ T^\forall\subseteq S^\forall\}$.
Then, the matching
\begin{align*}
\mu^*=&\{(d^e,h^*)\}_{e\in\hat{T}^\exists\cup\hat{T}^\forall}
    \cup\{(\bar{d}^e,h^e)\}_{e\in \hat{T}^\exists}
    \cup \{(d^e,h^e)\}_{e\in S^\exists\setminus\hat{T}^\exists}
    \cup\{(d_1,h^*),(d_2,h_1),(d_3,h_2),(d_4,h_1)\}
\end{align*}
is $1$-stable.

Conversely, consider the case when the instance is a \emph{no}-instance.
We show that there exists no $1$-stable matching in this case.
Suppose to the contrary that $\mu$ is a $1$-stable matching.
Then, $\mu$ must contain $(d_1,h^*)$
since the submarket induced by $\{d_1,d_2,d_3,d_4\}$ and $\{h_1,h_2\}$ 
is equivalent to Example~\ref{ex:nonexistence}.
Let $\tilde{T}^\exists\coloneqq \{e\in S^\exists\mid d^e\in \mu(h^*)\}$ and $\tilde{T}^\forall\subseteq S^\forall$ such that $\tilde{T}^\exists\cup T^\forall$ is a matching and $|\tilde{T}^\exists\cup\tilde{T}^\forall|=|X_1|$
(such a $\tilde{T}^\forall$ exists since the instance is a \emph{no}-instance).
Then $\{d^e\}_{e\in\tilde{T}^\exists\cup\tilde{T}^\forall}\cup\{d^*\}~(\in\cI_{h^*})$ is a $1$-blocking coalition for $h^*$
because $|\{d^e\}_{e\in \tilde{T}^\exists\cup\tilde{T}^\forall}\cup\{d^*\}|=|X_1|+1$ and $|\mu(h^*)|\le |X_1|$.
\end{proof}

\newtheorem*{thm:Sigma2Pc}{Theorem~\ref{thm:Sigma2Pc}}
\begin{thm:Sigma2Pc}
For any nonnegative real $\epsilon<1/2$,
it is \StP-hard to decide whether a given $(\Ucard,\Gvec{2,\epsilon})$-market has a $1$-stable matching.
\end{thm:Sigma2Pc}
\begin{proof}
We give a reduction from \EASSP{}.
Suppose that disjoint sets $S^\exists,S^\forall\subseteq X_1\times X_2\times X_3$, weights $a$, and an integer $q$ are given as an instance of \EASSP.
Without loss of generality, we may assume that $\sum_{e\in S^\exists}a(e)<q<\sum_{e\in S^\exists\cup S^\forall}a(e)$
and $|S^\forall|\ge 3$

We construct a $(\Usub,\Gcap)$-market that has a $1$-stable matching if and only if the given instance is a \emph{yes}-instance.
Let $M\coloneqq \sum_{e\in S^\exists\cup S^\forall}a(e)$ and $n\coloneqq |S^\exists\cup S^\forall|$.
Consider a market $(D,H,{\succ_D},u_H,\cI_H)$ with 
$D\coloneqq\{d_1,d_2,d_3,d_4\}\cup\{d_1^*,\dots,d_n^*\}\cup \{d^e\}_{e\in S^\exists\cup S^\forall}\cup \{\bar{d}^e\}_{e\in S^\exists}$ and 
$H\coloneqq \{h^*,h_1,h_2\}\cup \{h^e\}_{e\in S^\exists}$.
The doctors' preference over the acceptable hospitals are given as:
\begin{multicols}{2}
\begin{itemize}
\item $d^e\colon h^e,h^*$ $(e\in S^\exists)$,
\item $\bar{d}^e\colon h^e$ $(e\in S^\exists)$,
\item $d^e\colon h^*$ $(e\in S^\forall)$,
\item $d_i^*\colon h^*$ $(i\in\{1,\dots,n\})$,
\item $d_1\colon h^*,h_1,h_2$,
\item $d_2\colon h_1,h_2$,
\item $d_3\colon h_2,h_1$,
\item $d_4\colon h_2,h_1$.
\end{itemize}
\end{multicols}
The feasibility constraints are
$\cI_{h^e}\coloneqq \Gcap^{(1)}$ $(e\in S^\exists)$,
$\cI_{h_1}\coloneqq 2^{\{d_1,d_3\}}\cup 2^{\{d_2,d_4\}}$,
$\cI_{h_2}\coloneqq 2^{\{d_1,d_4\}}\cup 2^{\{d_2,d_3\}}$, and 
$\cI_{h^*}\in\Gvec{2,\epsilon}$ that can be represented by the following weights:
\begin{itemize}
\item $w(d^e,1)=(M+a(e))/(Mn+q)$, $w(d^e,2)=(M-a(e))/(Mn-q)$ $(e\in S^\exists\cup S^\forall)$,
\item $w(d_i^*,1)=M/(Mn+q)$, $w(d^e,2)=M/(Mn-q)$ $(i\in\{1,\dots,n\})$
\item $w(d_1,1)=2M/(Mn+q)$, $w(d_1,2)=0$.
\end{itemize}
Moreover, suppose that the utility of each hospital is cardinality.

Consider the case when the instance is a \emph{yes}-instance.
We show that there exists a $1$-stable matching in this case.
Let $\hat{T}^\exists\subseteq S^\exists$ be a certificate of the instance, i.e.,
$\sum_{e\in \hat{T}^\exists\cup T^\forall}a(e)\ne q$ for all $T^\forall\subseteq S^\forall$.
Then, the matching 
\begin{align*}
\mu^*=&\{(d^e,h^*)\}_{e\in\hat{T}^\exists}
       \cup\{(d_1^*,h^*),\dots,(d_{n-|\hat{T}^\exists|-2}^*,h^*)\}\\
       &\cup\{(\bar{d}^e,h^e)\}_{e\in \hat{T}^\exists}
       \cup\{(d^e,h^e)\}_{e\in S^\exists\setminus\hat{T}^\exists}
       \cup\{(d_1,h^*),(d_2,h_1),(d_3,h_2),(d_4,h_1)\}
\end{align*}
is $1$-stable.

Conversely, consider the case when the instance is a \emph{no}-instance.
We show that there exists no $1$-stable matching in this case.
Suppose to the contrary that $\mu$ is a $1$-stable matching.
Then, $\mu$ must contain $(d_1,h^*)$
since the submarket induced by $\{d_1,d_2,d_3,d_4\}$ and $\{h_1,h_2\}$ 
is equivalent to Example~\ref{ex:nonexistence}.
Let $\tilde{T}^\exists=\mu(h^*)\cap S^\exists$ and $\tilde{T}^\forall\subseteq S^\forall$ such that $\sum_{e\in \hat{T}^\exists\cup T^\forall}a(e)=q$.
(such a $\tilde{T}^\forall$ exists since the instance is a \emph{no}-instance).
Then, $\tilde{T}^\exists\cup\tilde{T}^\forall\cup\{d_1^*,\dots,d_{n-|\tilde{T}^\exists\cup\tilde{T}^\forall|}\}~(\in\cI_{h^*})$ is a $1$-blocking coalition for $h^*$
because $|\tilde{T}^\exists\cup\tilde{T}^\forall\cup\{d_1^*,\dots,d_{n-|\tilde{T}^\exists\cup\tilde{T}^\forall|}\}|=n$ and $|\mu(h^*)|\le n-1$.
\end{proof}

\section{Competitive Ratio of Algorithm~\ref{alg:vecgreedy}}\label{sec:vecgreedy}

\newtheorem*{thm:vecgreedy}{Theorem~\ref{thm:vecgreedy}}
\begin{thm:vecgreedy}
Algorithm~\ref{alg:vecgreedy} is $\rho$-competitive for the online $(\Ucard,\Gvec{\rho,\epsilon})$-packing problem
and $\rho/\epsilon$-competitive for the online $(\Uadd,\Gvec{\rho,\epsilon})$-packing problem.
\end{thm:vecgreedy}
\begin{proof}
Consider an instance of the online $(\Uadd,\Gvec{\rho,\epsilon})$-packing problem $(D,u,\cI)$ with
$D=\{d_1,\dots,d_n\}$ and $\cI=\{D'\subseteq D\mid \sum_{d\in D'}w(d,j)\le 1~(j\in [\rho])\}$.
We write $w(d)$ for the value of $\max_{j\in[\rho]} w(d,j)$.

Let $\sigma$ be an order on $D$ and $i\in [n]$.
We argue that $\rho\cdot u(\ALG(\sigma,i))\ge u(\OPT(\sigma,i))$ when the utility is cardinal and 
$(\rho/\epsilon)\cdot u(\ALG(\sigma,i))\ge u(\OPT(\sigma,i))$ when the utility is additive.
If $\sum_{t=1}^i w(\sigma(t))\le 1$, then we have $u(\ALG(\sigma,i))=u(\OPT(\sigma,i))$ by $\ALG(\sigma,i)=\{\sigma(1),\dots,\sigma(i)\}$.
Thus, we assume that $\sum_{t=1}^i w(\sigma(t))> 1$.

Without loss of generality, we assume that $\{\sigma(1),\dots,\sigma(i)\}=\{d_1,\dots,d_i\}$ and 
\[
  \frac{u(d_1)}{w(d_1)}\ge \frac{u(d_2)}{w(d_2)}\ge \dots \ge \frac{u(d_i)}{w(d_i)}.
\]
Let $\ell$ be the largest index such that $\sum_{t=1}^{\ell}w(d_t)\le 1$.
Note that $\ell<i$ and $\sum_{t=1}^{\ell+1}w(d_t)> 1$.
Then, by a simple induction, we can see that $\{d_{1},\dots,d_{\ell}\}\subseteq \ALG(\sigma,i)$.

For the cardinality case, we have $u(\ALG(\sigma,i))=|\ALG(\sigma,i)|=\ell$.
Note that $w(d_1)\le w(d_2)\le\dots\le w(d_n)$ in this case.
We claim that $u(\OPT(\sigma,i))\le \rho\cdot \ell$ by contradiction.
Suppose that $u(\OPT(\sigma,i))\ge \rho\cdot \ell+1$ (i.e., $|\OPT(\sigma,i)|\ge \rho\cdot \ell+1$).
Then, by the pigeonhole principle, $|\{d\in \OPT(\sigma,i)\mid w(d)=w(d,j)\}|\ge \ell+1$ for some $j^*\in [\rho]$.
For such a $j^*$, we have
\[
  \sum_{d\in \OPT(\sigma,i)}w(d,j^*)\ge \sum_{t=1}^{\ell+1}w(d_t)>1,
\]
which contradicts the feasibility of $\OPT(\sigma,i)$.
Hence, Algorithm~\ref{alg:vecgreedy} is $\rho$-competitive for the online $(\Ucard,\Gvec{\rho,\epsilon})$-packing problem. 

Moreover, for the additive case, $u(\ALG(\sigma,i))$ is at least
\begin{align*} 
\sum_{t=1}^{\ell}u(d_t) \ge \frac{\sum_{t=1}^{\ell}w(d_t)}{\sum_{t=1}^{\ell+1}w(d_t)}\left(\sum_{t=1}^{\ell+1}u(d_t)\right)
&= 
\left(1-\frac{w(d_{\ell+1})}{\sum_{t=1}^{\ell+1}w(d_{\sigma(t)})}\right)
\left(\sum_{t=1}^{\ell+1}u(d_{\sigma(t)})\right)\\
&\ge \epsilon\left(\sum_{t=1}^{\ell+1}u(d_{\sigma(t)})\right)\\
&\ge \epsilon\cdot\max\left\{u(D')\mid\! \begin{array}{l}w(D')\le 1,\\ D'\subseteq\{d_1,\dots,d_i\}\end{array}\!\!\right\}\\
&\ge \frac{\epsilon}{\rho}\cdot\max\left\{u(D')\mid\! \begin{array}{l}\sum_{d\in D'}w(d)\le \rho,\\ D'\subseteq\{d_1,\dots,d_i\}\end{array}\!\!\right\}\\
&\ge \frac{\epsilon}{\rho}\cdot\max\left\{u(D')\mid\! \begin{array}{l}\sum_{d\in D'}w(d,j)\le 1~(\forall j\in[\rho]),\\ D'\subseteq\{d_1,\dots,d_i\}\end{array}\!\!\right\}.
\end{align*}
Here, the first inequality holds since $\frac{\sum_{t=1}^{\ell}u(d_t)}{\sum_{t=1}^{\ell}w(d_{t})}$ is monotone nonincreasing for $\ell$
and the second inequality holds by $\sum_{t=1}^{\ell+1}w(d_t)>1$ and $w(d_{\ell+1})\le (1-\epsilon)$.
Thus, Algorithm~\ref{alg:vecgreedy} is $\frac{\rho}{\epsilon}$-competitive for the online $(\Uadd,\Gvec{\rho,\epsilon})$-packing problem. 
\end{proof}

\end{document}